\documentclass[a4paper,12pt]{article}

\usepackage[a4paper]{geometry}
\usepackage[english]{babel}

\usepackage{amsmath}
\usepackage{amsfonts}
\usepackage{amstext}
\usepackage{amssymb}
\usepackage{amsthm}
\usepackage{amscd}
\usepackage{mathrsfs}

\usepackage[pagebackref,draft=false]{hyperref}
\hypersetup{colorlinks,
linkcolor=myrefcolor,
citecolor=mycitecolor,
urlcolor=myurlcolor}

\usepackage[capitalize]{cleveref}
\usepackage{cite}
\usepackage{caption}
\usepackage{etaremune}

\usepackage{xcolor}
\definecolor{myurlcolor}{rgb}{0,0,0.4}
\definecolor{mycitecolor}{rgb}{0,0.5,0}
\definecolor{myrefcolor}{rgb}{0.5,0,0}
\usepackage{graphicx}
\usepackage{tikz}
\usepackage{tikz-cd}
\usetikzlibrary{cd}
\usetikzlibrary{decorations.pathmorphing,shapes}

\usepackage{etoolbox}
\usepackage{enumitem} 
\usepackage[]{latexsym}
\usepackage{braket}
\usepackage{caption}
\usepackage[utf8]{inputenx}
\usepackage[T1]{fontenc}
\usepackage{lmodern}
\usepackage{textcomp}
\usepackage{microtype}
\usepackage{totcount}
\usepackage{blindtext}

\newtheorem{theorem}{Theorem}[section]
\newtheorem{corollary}{Corollary}
\newtheorem{proposition}{Proposition}[section]
\newtheorem{definition}{Definition}[section]
\newtheorem{example}{Example}[section]

\newtheorem*{proof*}{Proof}

\newcommand{\be}{\begin{equation}}
\newcommand{\ee}{\end{equation}}
\newcommand{\bea}{\begin{eqnarray}}
\newcommand{\eea}{\end{eqnarray}}

\numberwithin{equation}{section}
\numberwithin{theorem}{section}

\title{Schwinger's picture of quantum mechanics: 2-groupoids and symmetries}
\date{}

\author{F. M. Ciaglia$^{1,7}$ \href{https://orcid.org/0000-0002-8987-1181}{\includegraphics[scale=0.7]{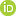}}, F. Di Cosmo$^{2,3,8}$ \href{https://orcid.org/0000-0003-0256-5913}{\includegraphics[scale=0.7]{ORCID.png}}, A. Ibort$^{2,3,9}$ \href{https://orcid.org/0000-0002-0580-5858}{\includegraphics[scale=0.7]{ORCID.png}}, \\ G. Marmo$^{4,5,10}$ \href{https://orcid.org/0000-0003-2662-2193}{\includegraphics[scale=0.7]{ORCID.png}}, L. Schiavone$^{3,4,6,11}$  \href{https://orcid.org/0000-0002-1817-5752}{\includegraphics[scale=0.7]{ORCID.png}} \\
\footnotesize{$^{1}$\textit{ Max Planck Institute for Mathematics in the Sciences, Leipzig, Germany}} \\
\footnotesize{$^{2}$\textit{ ICMAT, Instituto de Ciencias Matem\'{a}ticas (CSIC-UAM-UC3M-UCM)}} \\
\footnotesize{$^{3}$\textit{Depto. de Matem\'aticas, Univ. Carlos III de Madrid, Legan\'es, Madrid, Spain}} \\
\footnotesize{$^{4}$\textit{ INFN-Sezione di Napoli, Naples, Italy}} \\
\footnotesize{$^{5}$\textit{ Dipartimento di Fisica ``E. Pancini'', Universit\`a di Napoli Federico II,  Naples, Italy}} \\
\footnotesize{$^{6}$\textit{ Dipartimento di Matematica e Applicazioni "Renato Caccioppoli", Università di Napoli Federico II, Napoli, Italy}} \\
\footnotesize{$^{7}$\textit{ e-mail: \texttt{florio.m.ciaglia[at]gmail.com} and \texttt{ciaglia[at]mis.mpg.de}}} \\
\footnotesize{$^{8}$\textit{ e-mail: \texttt{fcosmo[at]math.uc3m.es}}} \\
\footnotesize{$^{9}$\textit{ e-mail: \texttt{albertoi[at]math.uc3m.es}}} \\
\footnotesize{$^{10}$\textit{ e-mail: \texttt{marmo[at]na.infn.it}}} \\ 
\footnotesize{$^{11}$\textit{ e-mail: \texttt{luca.schiavone[at]unina.it}}}  
}

\begin{document}

\maketitle

\begin{abstract}
Starting from the groupoid approach to Schwinger's picture of Quantum Mechanics, a proposal for the description of symmetries in this framework is advanced.
It is shown that, given a groupoid $G\rightrightarrows \Omega$ associated with a (quantum) system, there are two possible descriptions of its symmetries, one ``microscopic'', the other one ``global''.
The microscopic point of view leads to the introduction of an additional layer over the grupoid $G$, giving rise to a suitable algebraic structure of 2-groupoid.
On the other hand, taking advantage of the notion of group of bisections of a given groupoid, the global perspective allows to construct a group of symmetries out of a 2-groupoid.
The latter notion allows to introduce an analog of the Wigner's theorem for quantum symmetries in the groupoid approach to Quantum Mechanics.
\end{abstract}

\section{Introduction}

As a compromise between realistic descriptions and computability, abstract mathematical models have been introduced as approximated versions of physical systems in order to provide satisfactory explanations of observed phenomena and predictions of new ones. In such a scheme symmetries have immediately played a crucial role. Relativity, Quantum and Classical Mechanics, Statistical Mechanics and almost every area of physics benefits from the presence of symmetries, which in a certain sense put some order into the mathematical description.
Vaguely, a symmetry of a physical system could be defined as a family of ``transformations'' of the system preserving specific ``properties'' of it.
Depending on the explicit choice of the ``properties'' preserved, a finer distinction among symmetry transformations may be given in terms of dynamical or structural symmetries.

In this work we investigate how to rigorously introduce symmetries in Schwinger's picture of Quantum Mechanics (see \cite{Schwinger-2000, Sc51-1} for further details) which has been recently connected to the theory of groupoids by some of us \cite{C-DC-I-M-2020,C-DC-I-M-02-2020,C-I-M-02-2019,C-I-M-03-2019, C-I-M-05-2019,C-I-M-2018}.
As it will be shown in the rest of the paper, the groupoid approach, in contrast with Schwinger's original analysis which is close to the standard quantum mechanical treatment of symmetries based on the theory of (anti-)unitary transformations on Hilbert spaces (see for instance \cite[Sect. II]{Sc51-1}), permits a double description of symmetries, which we will call ``local'' or ``microscopic'' and ``global'' or ``macroscopic'', the latter directly related to the traditional modelisation of symmetries in terms of elements of a group.

It is gratifying to observe that the theory of groupoids, that was brought to a mathematical maturity in the smooth category by K. Mackenzie \cite{lie-groupoids}, could find, as a consequence of the results described in this paper and subsequent works, some of its most natural applications in the theory of symmetry of quantum systems, following closely the spirit of G.W. Mackey's oringinal contributions that pioneered the modern theory of groupoids \cite{Ma63,Ma66} as a natural continuation of the theory of groups and their representations.

Before diving into details, it is worth briefly recalling the basic ingredients of Schwinger's picture which are necessary to the developments of the paper.
We start assuming that the description of a quantum  system is given in terms of a groupoid $G\rightrightarrows \Omega$, where $\Omega$ is the space of objects.
The physical interpretation of elements
$$
x,y,z\cdots\,\in \Omega
$$
may be linked to the possible outcomes of a given set of observables (e.g. the outcomes of a Stern-Gerlach apparatus), and also to more abstract notions  as in the case of the double-slit experiment analysed in \cite{C-I-M-05-2019}.
Morphisms
$$
\alpha\,\colon\,x\,\rightarrow\,y
$$
may be interpreted as ``abstract amplitudes'' or ``virtual transitions'' among outcomes of experiments, and thus, in the rest of the paper, we will refer to them simply as ``transitions''.
Two maps,
\begin{equation}
s\,\colon\,G\rightarrow \Omega\,, \qquad t\,\colon\,G\rightarrow \Omega\,,
\end{equation}
called respectively the source and the target, associate every transition with its input and output objects.
For instance, $s(\alpha)= x$ and $t(\alpha)=y$.
Transitions, $\alpha\,\colon\,x\,\rightarrow\,y$ and $\beta\,\colon\,y'\,\rightarrow\,z$ can be composed, the product being denoted $\beta \circ \alpha$, whenever the source of $\beta$ coincides with the target of $\alpha$.
The  axioms satisfied by this composition law are:
\begin{itemize}
\item associativity;
\item for every outcome $x\in \Omega$ there is a transition, $1_x$ which leaves invariant any other composable transition, i.e.,
\begin{equation}
\alpha \circ 1_x = 1_y \circ \alpha = \alpha\,;
\end{equation}
\item for every transition $\alpha\,\colon \,x\rightarrow \, y$, there exists an inverse transition $\alpha^{-1}\,\colon\,y\,\rightarrow\,x$, such that $\alpha\circ \alpha^{-1}=1_y$ and $\alpha^{-1}\circ\alpha = 1_x$ (this property implements Feynman's principle of ``microscopic reversibility''\cite{feynman-thesis}).
\end{itemize}

The set of transitions $\alpha\,\colon\,x\,\rightarrow\,y$, with $x,\,y$ fixed will be denoted by $G(y,x)$ and the set of transitions $G(x,x)$ which form a group called the isotropy group at $x$, will be denoted by $G_x$.

According to modern algebraic formulations of quantum  (field) theories \cite{Haag-1996}, a central tool in the description of a quantum system is given by the  non-commutative algebra generated by the observables of the system.
In the groupoid picture under consideration, this algebra is realized as a suitable completion of the space  of compactly supported, complex-valued functions on $G$ equipped with an involution operator, say $\ast$, and a  non-commutative convolution product.
Under suitable conditions, for instance that the groupoid $G$ is a measure groupoid \cite{Ma63}, the space of compactly supported, complex-valued continuous functions on the groupoid, completed via a suitable norm, becomes a $C^*$-algebra \cite{Renault-1980, buneci} or a von Neumann algebra \cite{rings-operators, hahn}, referred to as the groupoid algebra of the groupoid, and denoted in what follows by $C^*(G)$,  and its real elements represent the observables of the system under investigation.
Then, as it happens in algebraic quantum field theory \cite{Haag-1996,Takesaki-2002},  the mathematical states  of this algebra (i.e., normalized, positive linear functionals) are associated with the physical states of the system.

In the following, we will focus on the purely algebraic aspects of the description of symmetries in the groupoid formalism.
This means that we will consider purely algebraic groupoids, and no assumptions on possible  topological and/or measure-theoretical properties of the groupoids will be made.
This choice is dictated by our wish, at this moment, to highlight the conceptual role of symmetries in the groupoid formalism  without having to worry too much of the strictly technical issues that come into play when topological or measure groupoids are considered.
Of course, we plan to address these technical issues in future publications.

It is worth stressing that, despite the simplifying  technical assumption on the nature of the groupoids mentioned before, the formalism is ``powerful'' enough to present relevant results.
For instance, in the case of discrete groupoids, it is possible to show that every state $\rho$ on the groupoid algebra determines   a positive definite function $\varphi_{\rho}$ on the groupoid by setting
\begin{equation}
\rho (\delta_{\alpha}) := \varphi_{\rho}(\alpha)\,,
\end{equation}
where $\delta_{\alpha}$ is the function on $G$ which takes the value $1$ only at the transition $\alpha$, and a natural connection between (a suitable family of) these states  and Sorkin's quantum measures \cite{Sorkin-1994}   has been shown in \cite{C-I-M-05-2019}.

As previously mentioned the notion of ``transformation''  is inherent to the notion of symmetry.
In our previous papers on groupoids and the foundations of Quantum Mechanics \cite{C-I-M-02-2019, C-I-M-03-2019}, a proposal for a theory of transformations of quantum systems was advanced.
It was argued there that a theory of transformations of quantum systems should be implemented as a set of rules $\xi,\,\zeta,\,\cdots$ describing further relations or ``substitutions'' among transitions.
In other words, the substitution $\xi\,\colon\,\alpha\,\Rightarrow \,\beta$, means that the actual transition $\alpha$ can be either transformed in the transition $\beta$ (for instance by modifying the experimental setting) or replaced formally in the description of the theory. In both situations the study of the possible substitutions of the given system could provide a helpful insight in its structure.
Actually much more is true.
Many of the structural results on physical theories come from imposing certain symmetries on them.
As it will be discussed in what follows, the natural way to introduce the notion of ``symmetry'' in the groupoid framework is by means of an appropriate collection of substitutions satisfying a few consistency axioms.
Such ``symmetries'' will impose serious restrictions on the system as it will be required that either the dynamics, or the states of the system, or both will be invariant.

Thus, the paper will be organized as follows.  Section 2 will be devoted to the discussion of the notion of symmetry by means of the theory of substitutions. It will be shown that the natural axioms for a family of substitutions makes them a 2-groupoid and, in this context, a 2-groupoid $S \rightrightarrows G \rightrightarrows \Omega$ intended to describe the symmetry of a quantum system will be called a \textit{symmetroid}.   In Sect. 3, a structure theorem for groupoids will be proved closely relating them to the theory of group actions, that is, it will be shown that any groupoid is isomorphic to a quotient of the action groupoid defined by the group of bisections of the given groupoid acting on itself.  Then, in Sect. 4, these ideas will be used to show that any symmetroid describes a group of (anti-)automorphisms of the given groupoid, and that this allows to define the canonical group of symmetries of the theory.   Such construction could be understood as the abstract setting of Wigner's theorem \cite{wigner} in the groupoidal description of quantum systems.   Finally a few simple examples that will illustrate the many aspects of the theory will be exhibited.

\section{Symmetries: The ``microscopic'' point of view}
\subsection{``Microscopic'' description of symmetries}

The features of the groupoid picture lead us to take firstly a ``microscopic'', or extremely detailed, approach to the notion of symmetry.
Specifically, given a groupoid $G\rightrightarrows\Omega$, the notion of transformation relevant to discuss the notion of symmetry, is implemented by means of a family of substitutions, where a substitution $\xi$ changes the transition $\alpha$ into the transition $\beta$; it will be denoted $\xi\colon \alpha\Rightarrow\beta$, with $\alpha,\beta\in G$, and diagrammatically  pictured as:
\begin{equation}
\begin{tikzcd}[column sep=large]
& \alpha \arrow[dd, "\;\xi", Rightarrow] \arrow[dr, bend left = 30] & \\
\bullet \arrow[ur, dash, bend left = 30] \arrow[dr, dash, bend right = 30] &  & \bullet \\
& \beta \arrow[ur, bend right = 30] &
\end{tikzcd}
\end{equation}

Referring to what is mentioned in the introduction about the distinction between structural and dynamical symmetries, a dynamical symmetry  will consist of a family of substitutions   leaving invariant a given state $\rho$ (or a family of them), or leaving invariant a given Hamiltonian $h\,\colon\,G\,\rightarrow\,\mathbb{C}$.
In the second case $h$ is invariant under the substitution $\xi\,\colon\,\alpha\,\Rightarrow\,\beta$ if $h(\alpha)= h(\beta)$.
On the other hand, we will say that $\rho$ is invariant under the substitution $\xi\,\colon\,\alpha\,\Rightarrow\,\beta$ if the associated function is invariant, i.e., $\varphi_{\rho}(\alpha)=\varphi_{\rho}(\beta)$.


The minimal set of axioms that a family of substitutions $S$ must satisfy to be considered a ``bona fide'' implementation of the notion of symmetry is the following:
\begin{itemize}
\item they can be composed, i.e., there exists a composition law, $\circ_{V}$, called ``vertical composition'', such that if
$$
\xi\,\colon\,\alpha \,\Rightarrow \,\,\beta\,,\quad \zeta\,\colon\,\beta\,\Rightarrow\,\gamma
$$
are substitutions in $S$, then there exists
$$
\zeta\circ_V\xi\,\colon\,\alpha\,\Rightarrow\,\gamma
$$
and $\zeta\circ_V\xi\,\in \,S$.
\item The composition law is associative.
\item There are ``trivial'' substitutions, i.e., substitutions $1_{\alpha},\, 1_{\beta}$, etc., such that
\begin{equation*}
\xi\circ_V 1_{\alpha}=\xi\,,\quad1_{\beta}\circ_V \xi = \xi\,.
\end{equation*}
\item Given any substitution $\xi\,\colon\,\alpha\,\Rightarrow\,\beta \: \in\, S$, such substitution can be ``undone'', i.e., there is another substitution $\xi^{-1}\,\in \,S$, $\xi^{-1}\,\colon\,\beta\,\Rightarrow\,\alpha$, such that
\begin{equation*}
\xi^{-1}\circ_V\xi = 1_{\alpha}\,,\quad\xi\circ_V\xi^{-1} = 1_{\beta}\,.
\end{equation*}
\end{itemize}
The vertical composition of substitutions will be denoted diagramatically as:
\begin{equation}
\begin{tikzcd}
 & \alpha \arrow[d, "\xi", Rightarrow] \arrow[dr, bend left = 30] & \\
\bullet \arrow[r, dash] \arrow[ur, dash, bend left = 30] \arrow[dr, dash, bend right = 30] & \beta \arrow[d, "\zeta", Rightarrow] \arrow[r] & \bullet \\
& \gamma \arrow[ur, bend right = 30] &
\end{tikzcd} \,\simeq \, \begin{tikzcd}
 & \alpha \arrow[dd, "\zeta \circ_V \xi", Rightarrow] \arrow[dr, bend left = 30] & \\
\bullet \arrow[ur, dash, bend left = 30] \arrow[dr, dash, bend right = 30] &  & \bullet \\
& \gamma \arrow[ur, bend right = 30] &
\end{tikzcd}
\end{equation}

The family of substitutions $S$, together with obvious source and target maps
$$
s,\,t\,\colon\,S\,\rightrightarrows G\,,\quad s(\xi)= \alpha\,,\quad t(\xi)=\beta\,,
$$
where $\xi\,\colon\,\alpha\,\Rightarrow\,\beta$, defines a groupoid structure on $G$. There is, however, an additional structure present on $S\,\rightrightarrows\, G$ that has its origin in the groupoid structure already present on $G\,\rightrightarrows\,\Omega$, and that will be called the horizontal composition $\circ_H$. Given two substitutions $\xi\,\colon\,\alpha\,\Rightarrow\,\beta$ and $\xi' \,\colon\, \alpha'\,\Rightarrow\,\beta'$ such that $\alpha,\,\alpha'$ and $\beta,\,\beta'$ are composable on $G$, i.e.,
$$
\alpha\,\colon\,x\,\rightarrow\,x'\,,\quad \alpha'\,\colon\,x'\,\rightarrow\,x''
$$
and
$$
\beta\,\colon\,y\,\rightarrow\,y'\,,\quad \beta'\,\colon\,y'\,\rightarrow\,y''\,,
$$
then we define the substitution
$$
\xi'\circ_H \xi \,\colon\,\alpha' \circ \alpha\,\Rightarrow\,\beta'\circ\beta\,.
$$
It will be denoted diagramatically as:
\begin{equation}
\begin{tikzcd}
 & \alpha \arrow[dd, "\xi", Rightarrow] \arrow[dr, bend left = 30] &  & \alpha' \arrow[dr, bend left = 30] \arrow[dd, "\xi'", Rightarrow] & \\
\bullet \arrow[ur, dash, bend left = 30] \arrow[dr, dash, bend right = 30] &  & \bullet \arrow[ur, dash, bend left = 30] \arrow[dr, dash, bend right = 30] & & \bullet  \\
& \beta \arrow[ur, bend right = 30] & & \beta' \arrow[ur, bend right = 30]
\end{tikzcd} \,\simeq\,
\begin{tikzcd}
 & \alpha' \circ \alpha \arrow[dd, "\xi' \circ_H \xi", Rightarrow] \arrow[dr, bend left = 30] & \\
\bullet \arrow[ur, dash, bend left = 30] \arrow[dr, dash, bend right = 30] &  & \bullet \\
& \beta' \circ \beta \arrow[ur, bend right = 30] &
\end{tikzcd}
\end{equation}

The horizontal composition satisfies some obvious axioms:
\begin{itemize}
\item associativity;
\item existence of units. Given
\begin{equation*}
\xi\,\colon\,\alpha\,\Rightarrow\,\beta\; \mathrm{with}\;\alpha\,\colon\,x\,\rightarrow\,x'\;\mathrm{and}\;\beta\,\colon\,y\,\rightarrow\,y'\,,
\end{equation*}
we consider the substitutions
\begin{equation*}
1_{y,x}\,\colon\,1_x\,\Rightarrow\,1_y\; \mathrm{and}\; 1_{y',x'}\,\colon\,1_{x'}\,\Rightarrow\,1_{y'}\,,
\end{equation*}
such that
\begin{equation*}
\xi\circ_H 1_{y,x}=\xi \,,\quad 1_{y',x'} \circ_H \xi = \xi\,.
\end{equation*}
\item existence of inverses, i.e., given $\xi\,\colon\,\alpha\,\Rightarrow\,\beta$, there is another substitution $\xi^{-H}\,\colon\,\alpha^{-1}\Rightarrow \, \beta^{-1}$, such that
\begin{equation*}
\xi^{-H}\circ_H\xi = 1_{y,x}\; \mathrm{and}\; \xi\circ_H\xi^{-H}=1_{y', x'}\,.
\end{equation*}
\end{itemize}

The two composition rules, the horizontal and the vertical ones, satisfy an additional compatibility condition, called exchange identity:
\begin{equation}
\left( \zeta\circ_V\xi \right)\circ_H\left( \zeta'\circ_V \xi' \right) = \left( \zeta'\circ_H\zeta \right)\circ_V\left( \xi'\circ_H \xi \right)\,,
\end{equation}
which is diagramatically pictured as follows
\begin{equation} \label{Eq: compatibility}
\begin{tikzcd}
 & \alpha \arrow[dd, "\zeta \circ_V \xi", Rightarrow] \arrow[dr, bend left = 30] &  & \alpha' \arrow[dr, bend left = 30] \arrow[dd, "\zeta' \circ_V \xi'", Rightarrow] & \\
\bullet \arrow[ur, dash, bend left = 30] \arrow[dr, dash, bend right = 30] &  & \bullet \arrow[ur, dash, bend left = 30] \arrow[dr, dash, bend right = 30] & & \bullet  \\
& \gamma \arrow[ur, bend right = 30] & & \gamma' \arrow[ur, bend right = 30]
\end{tikzcd} \,\simeq\, \begin{tikzcd}
 & \alpha \circ \alpha' \arrow[d, "\xi' \circ_H \xi", Rightarrow] \arrow[dr, bend left = 30] & \\
\bullet \arrow[r, dash] \arrow[ur, dash, bend left = 30] \arrow[dr, dash, bend right = 30] & \beta'\circ \beta \arrow[d, "\zeta' \circ_H \zeta", Rightarrow] \arrow[r] & \bullet \\
& \gamma' \circ \gamma \arrow[ur, bend right = 30] &
\end{tikzcd}
\end{equation}
Both sides of \eqref{Eq: compatibility} are equivalent to
\begin{equation}
\begin{tikzcd}
 & \alpha \arrow[d, "\xi", Rightarrow] \arrow[dr, bend left = 30] & & \alpha' \arrow[dr, bend left = 30] \arrow[d, "\xi'", Rightarrow] & \\
\bullet \arrow[r, dash] \arrow[ur, dash, bend left = 30] \arrow[dr, dash, bend right = 30] & \beta \arrow[d, "\zeta", Rightarrow] \arrow[r] & \bullet \arrow[ur, dash, bend left = 30] \arrow[dr, dash, bend right = 30] \arrow[r, dash] & \beta' \arrow[d, "\zeta'", Rightarrow] \arrow[r] & \bullet \\
& \gamma \arrow[ur, bend right = 30] & & \gamma' \arrow[ur, bend right = 30] &
\end{tikzcd}
\end{equation}

The structure $S\,\rightrightarrows\,G\,\rightrightarrows\,\Omega$ satisfying all previous axioms is called a $2$-groupoid and is a particular instance of a $2$-category (the first generalization of the notion of category coming from homotopy theory \cite{2_category_homotopy}).
Note that $S\,\rightrightarrows\, G$ is not a groupoid with respect to the horizontal composition $\circ_H$. Actually the source and target maps $s,t \colon (S, \circ_H) \to G$, are homomorphisms (covariant functors),
\begin{equation}
s(\xi' \circ \xi) = s(\xi') \circ s(\xi) = \alpha' \circ \alpha\,,
\end{equation}
for any $\xi \colon \alpha \Rightarrow \beta$, $\xi' \colon \alpha' \Rightarrow \beta'$, and thus  $S\,\rightrightarrows\, G\,\rightrightarrows\,\Omega$ is not a double groupoid (see \cite{brown} for the definition).

We conclude this first discussion by saying that the ``microscopic'' description of a symmetry in the groupoid picture of Quantum Mechanics is provided by a $2$-groupoid $S\,\rightrightarrows\,G\,\rightrightarrows\,\Omega$: Such structure will be called in what follows a ``symmetroid'' of the theory.

\subsection{The canonical symmetroid associated with a groupoid}

Since a groupoid is defined in terms of transitions which may be read as transformations between objects, it is reasonable to ask if there exists a built-in symmetry theory at the level  of the groupoid $G\,\rightrightarrows\,\Omega$, describing a given quantum system.
The answer is in the affirmative and we will now introduce two symmetroids which are naturally associated with a groupoid.

For the sake of simplicity we will start from the construction of the canonical little symmetroid,
$$
S_0\,\rightrightarrows\,G\,\rightrightarrows \,\Omega
$$
which is a subgroupoid of a larger groupoid which will be called the canonical symmetroyd and denoted
$$
S(G)\,\rightrightarrows\,G\,\rightrightarrows\,\Omega
$$
This latter notion will play a relevant role in the rest of the paper.

Let us start with a quantum system described by a groupoid $G\,\rightrightarrows \,\Omega$. For any $\gamma\in G_x$, i.e., $\gamma\,\colon\,x\,\rightarrow\,x$, consider the following family of substitutions
\begin{equation}
\xi^R_{\gamma}\,\colon\,\alpha\,\Rightarrow\,\beta=\alpha\circ\gamma^{-1}\,,
\end{equation}
i.e., we replace $\alpha\,\colon\,x\,\rightarrow\,y$ with its composition on the right with $\gamma^{-1}$. Similarly, if $\delta\in G_y$, we define
\begin{equation}
\xi_{\delta}^L\,\colon\,\alpha\,\Rightarrow\,\beta = \delta\circ\alpha\,.
\end{equation}
Notice that the substitutions $\xi_{\gamma}^R,\,\xi_{\delta}^L$ transform transitions in $G(y,x)$ into themselves, i.e., $\xi_{\gamma}^R,\,\xi_{\delta}^L$ map $G(y,x)$ into $G(y,x)$.

The family of substitutions $\xi^L_{\delta},\,\xi^R_{\gamma}$ satisfies the axioms of a $2$-groupoid over $G$. The vertical composition is defined in the obvious way, i.e.,
$$
\xi^{\cdot}_{\gamma}\circ_V \xi^{\cdot}_{\gamma'} = \xi^{\cdot}_{\gamma\circ\gamma'}\,,
$$
where $\cdot$ stands both for $R,\,L$. Clearly
\begin{equation}
\left( \xi^{\cdot}_{\gamma}\circ_V \xi^{\cdot}_{\gamma'} \right) \circ_V \xi^{\cdot}_{\gamma''} = \xi^{\cdot}_{\left( \gamma\circ\gamma'\right)\circ\gamma''} = \xi^{\cdot}_{ \gamma\circ \left( \gamma'\circ\gamma''\right)} =  \xi^{\cdot}_{\gamma}\circ_V \left( \xi^{\cdot}_{\gamma'} \circ_V \xi^{\cdot}_{\gamma''} \right)\,.
\end{equation}
Moreover,
\begin{equation}
\left( \xi^{R}_{\gamma^{-1}}\circ \xi^{R}_{\gamma} \right) \left( \alpha \right) = \xi^R_{1_x} (\alpha) = \alpha = \xi^L_{1_y}(\alpha) = \left( \xi^{L}_{\delta^{-1}}\circ \xi^{L}_{\delta} \right) \left( \alpha \right) \,,
\end{equation}
then
\begin{equation}
\left( \xi^{\cdot}_{\gamma} \right)^{-1} = \xi^{\cdot}_{\gamma^{-1}}\,.
\end{equation}
Finally,
\begin{equation}
\left( \xi^{\cdot}_{\gamma'}\circ_V \xi^{\cdot}_{\gamma} \right) \circ_H \left( \xi^{\cdot}_{\eta'}\circ_V \xi^{\cdot}_{\eta} \right) = \left( \xi^{\cdot}_{\gamma'\circ\gamma} \right) \circ_H \left( \xi^{\cdot}_{\eta'\circ\eta} \right)
\end{equation}
which implies
\begin{equation}
\begin{split}
\left( \xi^{\cdot}_{\gamma'}\circ_V \xi^{\cdot}_{\gamma} \right) \circ_H \left( \xi^{\cdot}_{\eta'}\circ_V \xi^{\cdot}_{\eta} \right) \left( \beta \circ \alpha \right) = \\
= \xi^{\cdot}_{\eta'\circ\eta}\left( \beta \right) \circ \xi^{\cdot}_{\gamma'\circ \gamma}\left( \alpha \right) \,.
\end{split}
\end{equation}
Furthermore,
\begin{equation}
\begin{split}
\left( \xi^{\cdot}_{\eta'}\circ_H \xi^{\cdot}_{\gamma'} \right) \circ_V \left( \xi^{\cdot}_{\eta}\circ_H \xi^{\cdot}_{\gamma} \right) \left( \beta \circ \alpha \right) = \\
= \left( \xi^{\cdot}_{\eta'}\circ_H \xi^{\cdot}_{\gamma'} \right) \left( \xi^{\cdot}_{\eta}\left(  \beta \right)\circ \xi^{\cdot}_{\gamma}\left( \alpha \right) \right)\,,
\end{split}
\end{equation}
which agrees with the previous one so that the exchange identity is satisfied.

However, this is not the end of the story, since there is another family of substitutions that will play an instrumental role in what follows. We call natural inversions the substitutions
\begin{equation}
\tau_{\alpha}\,\colon\,\alpha\,\Rightarrow\,\alpha^{-1}\,,
\end{equation}
any of which reverses the corresponding arrow $\alpha\,\colon\,x\,\rightarrow\,y$. Notice that
\begin{equation}
\tau_{\gamma\circ\alpha}\,\colon\,\gamma\circ\alpha \,\Rightarrow\, \alpha^{-1}\circ\gamma^{-1}\,,
\end{equation}
i.e.,
\begin{equation}
\tau_{\gamma\circ\alpha} = \xi^{R}_{\gamma}\circ_V \tau_{\alpha}\,.
\end{equation}
Then, we define the canonical little symmetroid\footnote{To include $\tau_{\alpha}$ is necessary to relax the definition of $2$-groupoid, in such a way that the source and target maps can be antihomomorphisms.} of $G\,\rightrightarrows\,\Omega$ as the $2$-groupoid $S_0(G) \,\rightrightarrows\,G\,\rightrightarrows\,\Omega$ generated by $\xi^{\cdot}_{\gamma},\,\tau_{\alpha}$, with $\gamma\in G_x$ and $G\ni \alpha\,\colon\,x\,\rightarrow\,y$.

An important observation is that the $2$-groupoid $S_0\,\rightrightarrows \,G$ is not connected, in fact the source and target of $\xi^{\cdot}_{\gamma}$ is $\alpha$ and $\xi^{\cdot}_{\gamma}\left( \alpha \right)$, so they lie within $G(y,x)$, whereas for $\tau_{\alpha}$, we have that $s\left( \tau_{\alpha} \right) = \alpha$ and $t\left( \tau_{\alpha} \right) = \alpha^{-1}$. Then we conclude that the orbits of $S_0$ are the subsets $G(x,\,y)\cup G(y,\, x)\subset G$.

Then, there is a larger symmetroid $S(G) \,\rightrightarrows\,G\,\rightrightarrows\,\Omega$, associated with $G$ which is given by all ``substitutions''
\begin{equation}
\xi^R_{\alpha_1}\,\colon\,\beta\,\Rightarrow \,\beta\circ\alpha^{-1}_1\,,
\end{equation}
whenever $t\left( \alpha_1 \right) = s\left( \beta \right)$, and
\begin{equation}
\xi^L_{\alpha_2}\,\colon\,\beta\,\Rightarrow \,\alpha_2\circ\beta\,,
\end{equation}
with $t\left( \beta \right) = s\left( \alpha_2 \right)$. Here
$$
\alpha_1\,\colon\,x\,\rightarrow \,y \quad \beta\,\colon\, x \,\rightarrow\,z \quad \alpha_2\,\colon\,z\,\rightarrow\,w
$$
are generic transitions in $G$. Exploiting the previous results about the canonical little $2$-groupoid $S_0(G)$, it is a routine computation to check the axioms of a $2$-groupoid for the groupoid generated by the family
\begin{equation}
\left\lbrace \xi^L_{\alpha},\, \xi^R_{\alpha}, \,\tau_{\alpha} \right\rbrace\,.
\end{equation}
The $2$-groupoid $S(G)\,\rightrightarrows\,G\,\rightrightarrows\,\Omega$ will be called the canonical symmetroid of the groupoid $G\,\rightrightarrows\,\Omega$ and corresponds to the action of the groupoid $G\,\rightrightarrows\,\Omega$ on itself (on the right and on the left) as described for instance in \cite[Chap. 4]{I-R-2019}. This structure generalizes what happens for groups: Given a group, indeed, there is a ``natural'' group acting upon itself, which is just the group itself acting on the right and on the left.

The $2$-groupoid $S_0\,\rightrightarrows\,G$ is a subgroupoid of $S$. We may also consider the $2$-groupoid $\mathfrak{T}\,\rightrightarrows\,G\,\rightrightarrows\,\Omega$, defined by the inversions $\tau_{\alpha}$ (notice that $\tau_{1_x}\,\colon\,1_x\,\Rightarrow\,1_x$) and the units $1_{\alpha}\,\colon\,\alpha\,\Rightarrow\,\alpha$. Such $2$-subgroupoid $\mathfrak{T}\subset S_0 \subset S$ is the ``microscopic reversibility'' symmetroid.

\section{Microscopic vs Global: The group of bisections.}

In order to introduce the notion of group of symmetries associated with a symmetroid, a preliminary structure needs to be defined. In this section it will be shown that, given a groupoid $G\,\rightrightarrows\,\Omega$, it is possible to build a group, whose elements will be called bisections, and which acts in a natural way on the space of objects $\Omega$.
Then, it will be proved that the original groupoid is isomorphic to a quotient of the action groupoid defined by the group of bisections acting on the space of objects.
This result tells us that one can adopt a global description of a groupoid via its group of bisections, or a local description via its transitions.
In particular, when the groupoid under consideration is the groupoid $S\rightrightarrows G$ which is the upper layer of the symmetroid $S\rightrightarrows G \rightrightarrows \Omega$, these points of view lead to the global and microscopic description of  symmetries mentioned in the introduction.


Let us start with a groupoid $G\,\rightrightarrows \,\Omega$. A bisection $b$ of $G\,\rightrightarrows \,\Omega$ is a subset $b\subset G$ such that the restrictions of both, the source and the target maps, are bijections (see Fig.\ref{fig:bisection_1}), i.e.,
\begin{equation*}
\mathrm{if}\: s_b= s\mid_b\,, \quad \mathrm{then} \quad s_b\,\colon\,b\,\rightarrow\,\Omega \quad \mathrm{is\: a\: bijective\: map},
\end{equation*}
as well as $t_b = t\mid_b$. This means that there exist maps $b_s\,\colon\,\Omega\,\rightarrow\,G$ and $b_t\,\colon\,\Omega\,\rightarrow\,G$, such that
\begin{equation}
b=\mathrm{range}\,b_s = \mathrm{range}\,b_t\,.
\end{equation}
The element $\alpha\in b \subset G$ such that $s(\alpha)=x$, i..e.,
$$
b_s\left( x \right) = \alpha\,,
$$
and $ t\left( \alpha \right)=y$, i.e.,
$$
b_t(y)=\alpha\,,
$$
will be denoted either as $b_s(x)=b_t(y)$, or simply $b(x)=b(y)$ if there is no risk of confusion. It is clear that any bisection $b$ defines a bijective map $\varphi_b\,\colon\,\Omega\,\rightarrow\,\Omega$ as follows
\begin{equation}
\varphi_b(x)= t(b_s(x))\quad \varphi_b^{-1}(y) = s(b_t(y))\quad \forall x,y\in \Omega\,.
\end{equation}

\begin{figure}[htp]
    \begin{center}
    \includegraphics[width=\textwidth]{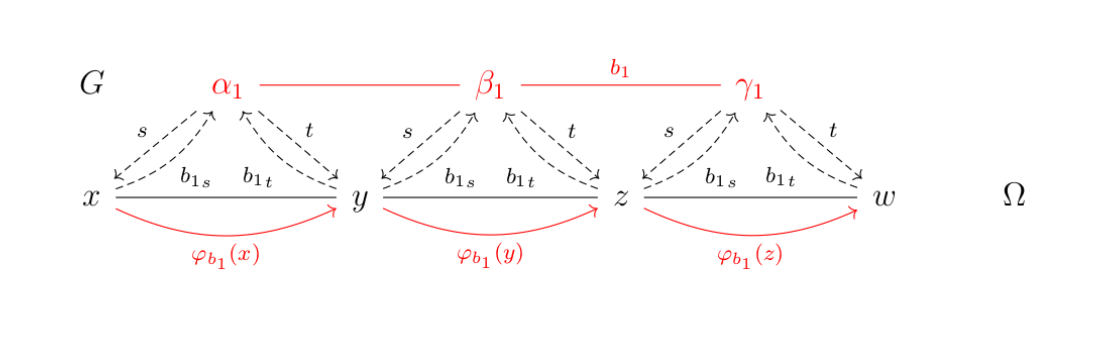}
    \caption{A schematic representation of a bisection: the red line connects the elements of the subset $b_1\subset G$ whilst the dotted arrows represent the map $b_s,b_t,s\mid_b,t\mid_b$. The red arrows denotes the bijective maps $\varphi_{b_1}$ associated with the bisection $b_1$.}
    \label{fig:bisection_1}
	\end{center}
\end{figure}

Even more interesting is the fact that bisections can be composed. Thus, given two bisections $b_1,\,b_2 \subset G$ we define its composition $b_2\circ b_1$ as the subset of $G$ obtained as the range of the maps $\left( b_2\circ b_1 \right)_s$, $\left( b_2\circ b_1 \right)_t$ defined as follows:
\begin{equation}\label{eq:composition_bisections}
\left( b_2\circ b_1 \right)_s (x) = (b_2)_s(y)\circ (b_1)_s(x)\,,
\end{equation}
and $y=t(b_1(s))=\varphi_{b_1}(x)$, or:
\begin{equation}
\left( b_2\circ b_1 \right)_t (z) = (b_2)_t(z)\circ (b_1)_t(y)\,,
\end{equation}
and $y=\varphi_{b_2}^{-1}(z)$. Then, the neutral element in this composition is given by the ``trivial'' bisection
\begin{equation}
b_e = \left\lbrace 1_x \mid x\in\Omega \right\rbrace\,,
\end{equation}
i.e., the subset of all units. Now, given a bisection $b$, its inverse $b^{-1}$ is given by the range of the maps
\begin{equation}
\left(b^{-1} \right)_s(x)=(b_t(x))^{-1}\,,\quad \left( b^{-1} \right)_t(y) = (b_s(y))^{-1}\,.
\end{equation}
Then, a straightforward computation shows that the set of all bisections, denoted in what follows by $\mathscr{G}$, satisfies all the axioms of a group.

We observe that the group $\mathscr{G}$ acts on $\Omega$, the action being the map $\mu \,\colon\,\mathscr{G}\times\Omega \,\rightarrow\,\Omega$,
\begin{equation}
\mu(b,x) = \varphi_b(x) =: b\cdot x\,.
\end{equation}
Indeed, we have that:
\begin{equation}
\begin{split}
\left( b_2\circ b_1 \right)(x) &= \varphi_{b_2\circ b_1}(x) = t\left( \left( b_2\circ b_1 \right)_s(x) \right) = t \left( \left( b_2 \right)_s(y)\circ \left(b_1\right)_s(x) \right) = \\
&= t\left( \left( b_2 \right)_s \left( \varphi_{b_1}(x) \right) \right) = \varphi_{b_2}\left( \varphi_{b_1} (x) \right) = b_2\cdot \left( b_1 \cdot x \right)\,.
\end{split}
\end{equation}
The following theorem will show us some additional relations between the groupoid $G\,\rightrightarrows \,\Omega$ and the action of the group $\mathscr{G}$ on $\Omega$.
\begin{theorem}\label{Th:reconstruction of groupoids}
Let $G\,\rightrightarrows \,\Omega$ be a connected groupoid and $\mathscr{G}$ the group of its bisections acting on $\Omega$. Then, $G\,\rightrightarrows\,\Omega$ is isomorphic, as groupoid, to the quotient of the action groupoid $\mathscr{G}\times\Omega\,\rightrightarrows\,\Omega$ by the normal subgroupoid $N = \left\lbrace \left( y;b;x \right)\mid \varphi_{b}(x)=x=y \right\rbrace$, where $\left( y;b;x \right)$, with $y=b\cdot x$, denotes an element of the groupoid $\mathscr{G}\times \Omega \,\rightrightarrows\,\Omega$.
\end{theorem}
\begin{proof}
It is convenient to denote the elements of the action groupoid $\mathscr{G} \times \Omega \rightrightarrows \Omega$ by triples $(y; b ; x)$, where $y = \varphi_b(x) = b \cdot x$, $x \in \Omega$, $b \in \mathscr{G}$. Then, the composition law of the action groupoid $\mathscr{G}\times\Omega \,\rightrightarrows\,\Omega$ is
\begin{equation}
\left( z;b_2;y \right)\circ\left( y;b_1;x \right) = \left( z; b_2\circ b_1 ; x \right)\,,
\end{equation}
whereas source and target maps are
\begin{equation}
s\left( y;b;x \right) = x\,,\quad t\left( y;b;x \right) = y\,.
\end{equation}
It follows that the units of the action groupoid are the elements:
\begin{equation}
\left( x; b_e; x \right)\,,
\end{equation}
where $b_e$ is the trivial bisection or neutral element of the group $\mathscr{G}$.

Define the map $A\,\colon\,\mathscr{G}\times\Omega\,\rightarrow\,G$
\begin{equation}
A\left( y;b;x \right) = b_s(x)= b_t(y)\,.
\end{equation}
Then, the map $A$ is a functor, i.e., a homomorphism of groupoids. Indeed,
\begin{equation}
\begin{split}
A\left( \left( z; b_2; y \right) \circ \left( y; b_1 ; x \right)\right) &= A\left( z; b_2\circ b_1 ; x \right) = \left( b_2\circ b_1 \right)_s(x) = \\
\left( b_2 \right)_s(y) \circ \left( b_1 \right)_s(x) &= A\left( z;b_2;y \right)\circ A\left( y;b_1;x \right)
\end{split}
\end{equation}
and
\begin{equation}
A(x;b_e;x) = 1_x\,.
\end{equation}
However, this groupoid homomorphism is not injective, since it has a non-trivial kernel:
\begin{equation}
N = \mathrm{Ker}A = \sqcup_{x\in \Omega}\mathrm{Ker}A(x)\,,\quad \mathrm{Ker}A(x)= \left\lbrace b\in\mathscr{G}\mid b_s(x)=1_x \right\rbrace\,,
\end{equation}
and $N$ is a normal subgroupoid of $\mathscr{G}\times\Omega\, \rightrightarrows\,\Omega$.

Let us prove, now, that $A$ is onto. It suffices to show that for any $\alpha \in G$, $\alpha\,\colon\,x\,\rightarrow\,y$, there is a bisection $b$ such that $b_s(x)=\alpha$. This can be seen considering the canonical projection
$$
\pi\,\colon\,G\,\rightarrow\,G(\Omega)=\Omega\times\Omega,\quad \quad \pi(\alpha)\,=\, \left( t(\alpha), \,s(\alpha) \right).
$$
Let us consider in $\Omega\times\Omega$ the graph of a bijective map $\varphi \,\colon\,\Omega\,\rightarrow\,\Omega$ such that $\varphi(x)=y$. Then, consider a section of $\pi$ restricted to $\pi^{-1}\left( \mathrm{graph}\,\varphi \right)$, i.e., a map
$$
\sigma\,\colon\,\mathrm{graph}\,\varphi\subset \Omega\times\Omega\,\rightarrow\,G \; \mathrm{such\: that} \;\;\pi\circ\sigma = \mathrm{id}\,.
$$
Therefore, $b= \sigma\left( \mathrm{graph}\,\varphi \right)$ is the bisection we were looking for.

Since $N$ is a normal subgroupoid, because of the first isomorphism theorem of groupoids (see, for instance \cite[Chap. 5]{I-R-2019}), we get the short exact sequence of groupoids and homomorphisms:
\begin{equation}
1\,\rightarrow\,N\,\rightarrow\,\mathscr{G}\times\Omega\,\rightarrow\,A\left( \mathscr{G}\times\Omega \right) \,\rightarrow\,1\,,
\end{equation}
with $A\left( \mathscr{G}\times\Omega \right) = \mathscr{G}\times\Omega / N$. Since $A$ is a surjective groupoid homomorphism, $A\left( \mathscr{G}\times\Omega \right)$ is isomorphic to $G\,\rightrightarrows\,\Omega$.
\end{proof}
An immediate consequence is the following corollary:
\begin{corollary}
Any connected algebraic groupoid $G\,\rightrightarrows\,\Omega$ is isomorphic to the quotient of an action groupoid.
\end{corollary}

A couple of remarks are in order now. Firstly, the previous structure theorem can fail in the category of Lie groupoids\footnote{In a subsequent work it will be shown that the theorem is still true in the category of measure groupoids considering an inessential restriction on the given groupoid.} because of the lack of smooth bisections with the desired properties. Second, the quotient of an action groupoid by a normal subgroupoid does not have to be an action groupoid anymore, i.e., not all groupoids are action groupoids (examples are provided easily using ``tiling'' groupoids (see \cite{I-R-2019, wein-1996})). Finally, if a groupoid is not connected, we can use the theorem, in the algebraic setting, on each one of its connected components.

In summary, given a symmetry of a quantum system described as the groupoid $G\,\rightrightarrows\,\Omega$, i.e., a $2$-groupoid $S\,\rightrightarrows\,G\,\rightrightarrows\,\Omega$, then such ``microscopic'' description of the symmetry can be seen as a global symmetry described by the group $\mathscr{S}$ of bisections of $S\,\rightrightarrows\,G$ acting on $G$ (it is important to recall here that $S\,\rightrightarrows\,G$ is a groupoid only with respect to the vertical composition and this is the composition law that will be used to define the product of two bisections). Thus, the symmetroid $S\,\rightrightarrows\,G$ will be recovered as a quotient of the action groupoid $\mathscr{S}\times G\,\rightrightarrows\,G$. We will analyze these issues in the next section.

\begin{example}{\textbf{Reconstruction of finite groupoids.}}\label{ex.reconstruction}
A simple example is now useful to illustrate the content of the theorem above. Let us consider the cyclic groupoid of order 2 over 2 elements $C_2(4) \,\rightrightarrows\, \Omega_2$ (see \cite{I-R-2019}). The space of objects consists of two elements, $\Omega_2 = \left\lbrace +, - \right\rbrace$ whereas the groupoid $C_2(4)$ contains 8 transitions:
\begin{equation}
\begin{split}
\alpha_1 = (+, \sigma, -) \,,\quad \alpha_2 = (+, e, -) \,, \quad  1_{\pm} = (\pm , e, \pm) \\
\beta_1 = (-, \sigma, +) \,,\quad \beta_2 = (-, e, +) \,,\quad \sigma_{\pm} = (\pm , \sigma, \pm)\,,
\end{split}
\end{equation}
where $\left\lbrace e, \sigma \right\rbrace$ are the elements of the group of permutations of 2 objects. In other words, the transition $\alpha_1$ can be interpreted as a transition between the two objects of the groupoid which contemporarily produces a permutation in an internal register, such that two consecutive permutations bring the register to the original configuration.

The group of bisections, $\mathcal{G}$, is made up of subsets of the groupoid $C_2(4)$ such that the restrictions of the source and target maps to these subsets are bijective. In this case, therefore, any bisection contains two transitions and we have 8 bisections:
\begin{equation}
\begin{split}
b_e = (1_+, 1_-)\quad b_+ = (\sigma_+, 1_-)\quad  b_- = (1_+, \sigma_-) \quad b_g = (\sigma_+ , \sigma_-)\\
b_1 = (\alpha_1, \beta_1)\quad b_2 = (\alpha_1, \beta_2) \quad b_3= (\alpha_2, \beta_1) \quad b_4= (\alpha_2, \beta_2)
\end{split}
\end{equation}
which define the multiplication Tab.\eqref{table_1}. A schematic representation of the set of bisections is contained in Fig.\ref{fig:bisection_2}.
\begin{table}[h!]
\centering
\begin{tabular}{|c|c c c c|c c c c|}
\hline
0 & $ b_e $ & $ b_+ $ & $b_-$ & $b_g$ & $b_1$ & $b_2$ & $b_3$ & $b_4$ \\
\hline
$b_e$ & $b_e$ & $b_+$ & $b_-$ & $b_g$ & $b_1$ & $b_2$ & $b_3$ & $b_4$ \\
$b_+$ & $b_+$ & $b_e$ & $b_g$ & $b_-$ & $b_3$ & $b_4$ & $b_1$ & $b_2$ \\
$b_-$ & $b_-$ & $b_g$ & $b_e$ & $b_+$ & $b_2$ & $b_1$ & $b_4$ & $b_3$ \\
$b_g$ & $b_g$ & $b_-$ & $b_+$ & $b_e$ & $b_4$ & $b_3$ & $b_2$ & $b_1$ \\
\hline
$b_1$ & $b_1$ & $b_2$ & $b_3$ & $b_4$ & $b_e$ & $b_+$ & $b_-$ & $b_g$ \\
$b_2$ & $b_2$ & $b_1$ & $b_4$ & $b_3$ & $b_-$ & $b_g$ & $b_e$ & $b_+$ \\
$b_3$ & $b_3$ & $b_4$ & $b_1$ & $b_2$ & $b_+$ & $b_e$ & $b_g$ & $b_-$ \\
$b_4$ & $b_4$ & $b_3$ & $b_2$ & $b_1$ & $b_g$ & $b_-$ & $b_+$ & $b_e$ \\
\hline
\end{tabular}
\caption{Multiplication table of the group $\mathscr{G}$ of bisections of the groupoid $C_2(4)\,\rightrightarrows \,\Omega_2$.}
\label{table_1}
\end{table}
It is immediate to notice that there is a normal subgroup
\begin{equation}
\mathscr{G}_0 = \left\lbrace b_e, \,b_+,\,b_-,\,b_g \right\rbrace
\end{equation}
which is isomorphic to the group $\mathscr{G}_0 \simeq \mathbb{Z}_2\times \mathbb{Z}_2$ (this group is also known as the Klein four-group). The quotient group $\mathscr{H} = \mathscr{G} / \mathscr{G}_0$ is isomorphic to the group $\mathbb{Z}_2$, the group of permutations of two elements. As a side remark, let us notice that the group $\mathbb{Z}_2$ is the group of bijective maps of the space of objects $\Omega_2$.
\begin{figure}[htp]
    \begin{center}
    \includegraphics[width=\textwidth]{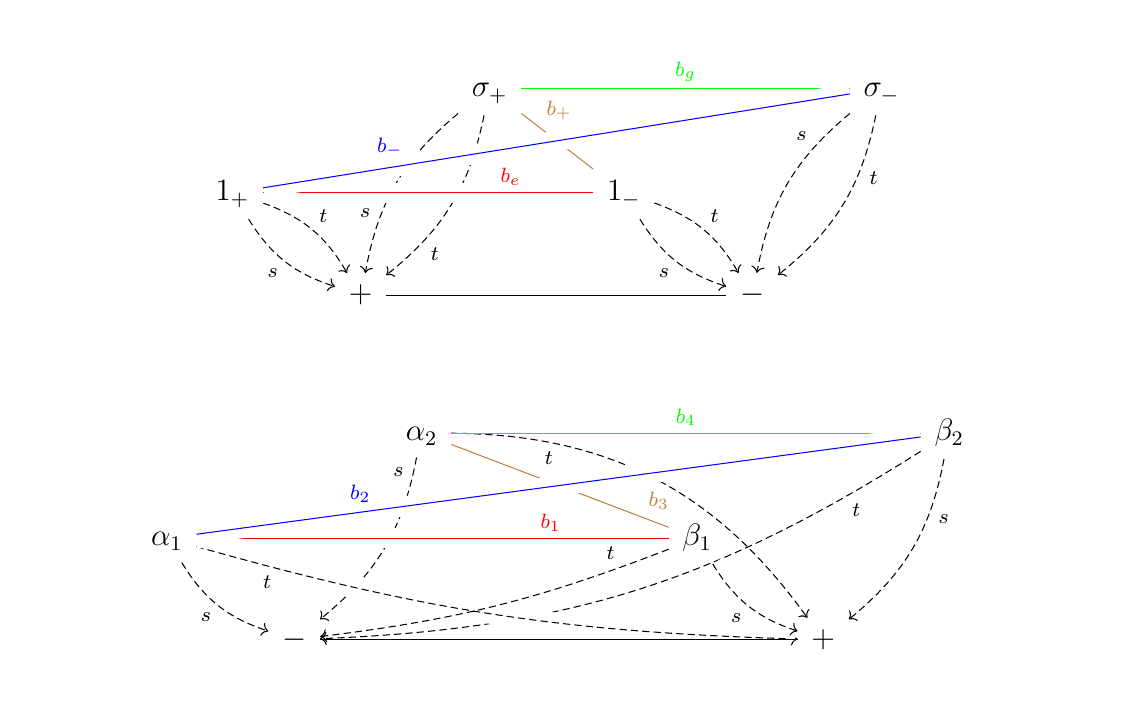}
    \caption{Schematic representation of the group of bisections $\mathscr{G}$ of the groupoid $C_2(4)$.}
    \label{fig:bisection_2}
	\end{center}
\end{figure}

A section of the projection $\pi \,\colon\,\mathscr{G}\,\rightarrow\,\mathscr{H}$ is given, for instance, by the map $\rho\,\colon\,\mathscr{H}\,\rightarrow\,\mathscr{G}$
\begin{equation}
\rho(e) = b_e\quad \rho(\sigma) = b_4\,,
\end{equation}
which is also a homomorphism of groups. This section defines a map between the group $\mathscr{H}$ and the group of automorphisms of $\mathscr{G}_0$, i.e., $W_{\cdot}^{\rho}\,\colon\,\mathscr{H}\,\rightarrow\, \mathrm{Aut}(\mathscr{G}_0)$
\begin{equation}
W_h^{\rho}(k) = \rho(h)k \left( \rho(h)\right)^{-1}\,.
\end{equation}
In the case considered in this example, we have:
\begin{equation}
\begin{split}
W^{\rho}_e &= id_K \\
W^{\rho}_{\sigma}(b_+) = b_-,\quad W^{\rho}_{\sigma}(b_-) = b_+&,\quad W^{\rho}_{\sigma}(b_e) = b_e,\quad W^{\rho}_{\sigma}(b_g) = b_g\,.
\end{split}
\end{equation}
Therefore, the group of bisections is an extension of the group $\mathscr{H}$ by the group $\mathscr{G}_0$ isomorphic to the semidirect product $\mathscr{G}_0\rtimes_{W^{\rho}} \mathscr{H}$.

According to Theorem \ref{Th:reconstruction of groupoids} the groupoid $C_2(4)$ is isomorphic to the quotient groupoid of the action groupoid $\mathscr{G}\times \Omega_2$ by a normal subgroupoid $N$. Concerning the situation presented in this example, the normal subgroupoid $N$ contains the following transitions of $\mathscr{G}\times \Omega_2$
\begin{equation}
N=\left\lbrace (+,b_-,+), (+,b_e,+), (-,b_+,-), (-,b_e,-) \right\rbrace\,.
\end{equation}
Then, the quotient space $\mathscr{G}\times \Omega_2  / N$ contains the following elements:
\begin{eqnarray*}
& \left[ +, b_e , + \right] = \left\lbrace (+,b_-,+), (+,b_e,+) \right\rbrace,\;\; \left[ +, b_+ , + \right] = \left\lbrace (+,b_+,+), (+,b_g,+) \right\rbrace \\
&\left[ -, b_e , - \right] = \left\lbrace (-,b_+,-), (-,b_e,-) \right\rbrace,\;\; \left[ -, b_- , - \right] = \left\lbrace (-,b_-,-), (-,b_g,-) \right\rbrace \\
&\left[ -, b_1 , + \right] = \left\lbrace (-,b_1,+), (-,b_3,+) \right\rbrace,\;\;\left[ -, b_4 , + \right] = \left\lbrace (-,b_4,+), (-,b_2,+) \right\rbrace \\
&\left[ +, b_1 , - \right] = \left\lbrace (+,b_1,-), (+,b_3,-) \right\rbrace,\;\;\left[ +, b_4 , - \right] = \left\lbrace (+,b_4,-), (+,b_2,-) \right\rbrace\,.
\end{eqnarray*}
A straightforward computation shows that this groupoid is isomorphic to the groupoid $C_2(4)$, the isomorphism being as follows:
\begin{equation}
\begin{split}
\left[ +, b_e , + \right] \leftrightarrow 1_+ ,\;\; \left[ +, b_+ , + \right] \leftrightarrow \sigma_+ ,\;\; \left[ -, b_e , - \right] \leftrightarrow 1_- ,\;\; \left[ -, b_- , - \right] \leftrightarrow \sigma_- \\
\left[ -, b_1 , + \right] \leftrightarrow \beta_1 ,\;\; \left[ -, b_4 , + \right] \leftrightarrow \beta_2 ,\;\; \left[ +, b_1 , - \right] \leftrightarrow \alpha_1 ,\;\; \left[ +, b_4 , - \right] \leftrightarrow \alpha_2\,.
\end{split}
\end{equation}
\end{example}
\vspace{0.5cm}
The above example is illustrating also a more general result about the structure of the group of bisections of connected discrete groupoids. Indeed, let us consider a discrete connected groupoid $G\rightrightarrows \Omega$ and its fundamental sequence (see \cite{I-R-2019}):
\begin{equation}
\mathbf{1}_{\Omega} \, \rightarrow \, G_0\, \rightarrow \,G\, \rightarrow \, G(\Omega)\,\rightarrow \, \mathbf{1}_{\Omega}\,,
\end{equation}
where $G_0 = \sqcup_{x\in \Omega} G_x$ is the isotropy normal subgroupoid of $G$ and $G(\Omega)$ is the pair groupoid over $\Omega$. This sequence of groupoid homomorphisms splits and we have
\begin{equation*}
G \cong G(\Omega) \times \Gamma\,,
\end{equation*}
where $\Gamma$ is a group isomorphic to any of the isotropy groups of $G$. Therefore, one can look for the group of bisections, say $\mathscr{G}$, of the direct product $G(\Omega) \times \Gamma$. For such a groupoid, a bisection is expressed via the following map:
\begin{equation}
\begin{split}
b_s \,\colon\, \Omega \, \rightarrow \, G(\Omega)\times \Gamma \\
b_s(x) = \left( (\psi(x), x) , \gamma(x)  \right)\,,
\end{split}
\end{equation}
where $g(x)\in \Gamma$ for any $x\in \Omega$, and
\begin{equation}
\psi\,\colon\,\Omega\,\rightarrow\,\Omega
\end{equation}
is a bijective map on the space of objects. The product of two bisections, $b_1,b_2\in\mathscr{G}$, is expressed as follows:
\begin{equation}
\begin{split}
(b_2\cdot b_1)_s(x) = (b_2)_s(\varphi_{b_1}(x))&\circ(b_1)_s(x) = \\
= \left((\psi_2(\psi_1(x)), \psi_1(x)),\,\gamma_2(\psi_1(x))\right) &\circ \left( (\psi_1(x), x),\, \gamma_1(x)  \right) = \\
\left( (\psi_2\circ\psi_1(x), x),\,\gamma_2 (\psi_1(x))\gamma_1(x)  \right)
\end{split}
\end{equation}
It is easy to see that the bisections of the type:
\begin{equation}
\gamma_s (x) = \left((x,x), \,\gamma(x)  \right)
\end{equation}
form a normal subgroup, say $\mathscr{G}_0$, from which one obtains the following short exact sequence of groups:
\begin{equation}
1 \, \rightarrow \, \mathscr{G}_0\,\rightarrow\,\mathscr{G}\,\rightarrow\,\mathscr{H} := \mathscr{G}/\mathscr{G}_0 \,\rightarrow\, 1 \,.
\end{equation}
The quotient group is isomorphic to the group of bijective maps of the space of objects and the sequence splits, for instance, via the section
\begin{equation}
\begin{split}
\rho\,\colon\,\mathscr{H}\,&\rightarrow\,\mathscr{G}\\
\rho(\psi) \, \rightarrow \,& \left( (\psi(\cdot), \cdot),\, e \right)\,.
\end{split}
\end{equation}
Therefore, the group of bisections is isomorphic to the semidirect product
\begin{equation}
\mathscr{G}\cong \mathscr{G}_0 \rtimes_{W^{\rho}} \mathscr{H}
\end{equation}
where $W^{\rho}$ is the following map:
\begin{equation}
\begin{split}
W^{\rho}_{\cdot}\,\colon\,\mathscr{H}\,&\rightarrow\,\mathrm{Aut}(\mathscr{G}_0)\\
\psi \in \mathscr{H} \,\rightarrow\, W_{\psi}^{\rho}(((x,x)&,\,\gamma(x))) = ((x,x),\gamma(\psi(x)))
\end{split}
\end{equation}

\section{The canonical group of symmetries of a quantum system.}

Given a quantum system $G\,\rightrightarrows\,\Omega$ with a symmetroid $S\,\rightrightarrows\,G\,\rightrightarrows\,\Omega$, we may consider the group of bisections $\mathscr{S}$ of the groupoid $S \rightrightarrows G$ which acts on $G$, the group composition law  given by, (\ref{eq:composition_bisections}):
\begin{equation*}
(b_2\cdot b_1)_s(\alpha) = (b_2)_s(\varphi_{b_1}(\alpha))\circ_V (b_1)_s(\alpha)\, , \qquad b_1,\,b_2\in \mathscr{S} \, .
\end{equation*}
However there is an issue concerning this definition, that is, the bisections in $\mathscr{S}$ do not necessarily respect the groupoid structure $G\,\rightrightarrows\,\Omega$.

Consider $\alpha,\,\alpha'\in G$ two composable transitions, i.e. $t(\alpha)=s(\alpha')$, and $b \in \mathscr{S}$ a bisection such that,
\begin{equation}
\left( b \right)_s(\alpha)\,\colon\,\alpha\,\Rightarrow\,\beta\,,\quad \left( b \right)_s(\alpha')\,\colon\,\alpha'\,\Rightarrow\,\beta'\,,
\end{equation}
and $\beta,\,\beta'$ are composable, too. Then, one may define the horizontal composition:
\begin{equation}
b_s(\alpha')\circ_H b_s(\alpha)\,\colon\, \alpha'\,\circ\,\alpha \, \Rightarrow \, \beta' \circ \beta\,,
\end{equation}
which is, generally, different from $b_s\left( \alpha' \circ \alpha \right)$.
Since the source map $s\,\colon\, S \,\rightarrow\,G$ is a horizontal homomorphism (or anti-homomorphism), i.e.,
$$
s(\xi' \circ_H \xi) = s(\xi')\circ s(\xi)\,,
$$
applying $s$ to $b_s(\alpha')\circ_H b_s(\alpha)$ and $b_s(\alpha' \circ\alpha)$ we get in both cases $\alpha' \circ \alpha$. Hence, it follows that
\begin{equation}
b_s\left( \alpha' \right)\circ_H b_s\left( \alpha \right) = b_s(\alpha'\circ \alpha)\circ_V \gamma \left( \alpha',\,\alpha \right)\,,
\end{equation}
where $\gamma\left( \alpha',\,\alpha \right) \in S_{\alpha'\circ\alpha}$ and $S_{\alpha'\circ\alpha}$ denotes the isotropy group of the transition $\alpha'\circ\alpha$ in $S$. Analogously one gets
\begin{equation}
b_s\left( \alpha' \right)\circ_H b_s\left( \alpha \right) = \tilde{\gamma} \left( \alpha',\,\alpha \right)\circ_V b_s(\alpha'\circ \alpha)\,,
\end{equation}
with $\tilde{\gamma}\left( \alpha',\,\alpha \right) \in S_{\beta' \circ \beta}$. These elements $\gamma(\alpha',\,\alpha)$ will define a $2$-cocycle on $G$ with values in the group bundle over $G$ determined by the isotropy groups $S_{\alpha'\circ\alpha}$, and, in certain cases, could be computed explicitly (for more details on groupoid cohomology see, for instance, \cite{I-R-2019,Renault-1980, lie-groupoids}). In what follows, we will restrict our attention to those bisections that preserve the composition of quantum transitions, that is, those bisections for which the cocycle $\gamma$ discussed before is trivial.
\begin{definition}\label{flat_bisection}
Given a symmetroid $S\,\rightrightarrows\,G\,\rightrightarrows\,\Omega$, a bisection $b\subset S$ such that
\begin{equation}
b_s\left( \alpha' \right)\circ_H b_s\left( \alpha \right) = b_s\left( \alpha'\circ\alpha \right) \left( \mathrm{or}\; b_s\left( \alpha\circ \alpha' \right)\,\right)\,,
\end{equation}
for every pair of composable transitions $\alpha,\,\alpha'\in G$, will be called flat or horizontal.
\end{definition}

It is clear that the set of flat bisections determines a subgroup $\mathscr{S}^{\flat}\subset \mathscr{S}$.
\begin{proposition}
Let $S\,\rightrightarrows\,G\,\rightrightarrows\,\Omega$ be a $2$-groupoid, then the subset of flat bisections $b\in \mathscr{S}$ defines a subgroup of $\mathscr{S}$, denoted $\mathscr{S}^{\flat}$. Moreover the action of $\mathscr{S}^{\flat}$ on $G$ is by (anti-)automorphisms of $G$.
\end{proposition}
\begin{proof}
Let $b\in \mathscr{S}^{\flat}$, then the associated map $\varphi_b\,\colon\,G\,\rightarrow\,G$ is an (anti-)automorphism of $G$:
\begin{equation}
\begin{split}
\varphi_b\left( \beta \circ \alpha \right) = t\left( b_s\left( \beta\circ \alpha \right) \right) &= t\left( b_s\left( \beta \right) \circ_H b_s\left( \alpha \right)  \right) = \\
= t\left( b_s\left( \beta \right) \right) \circ t\left( b_s\left( \alpha \right) \right) &= \varphi_b\left( \beta \right) \circ \varphi_b\left( \alpha \right)\,.
\end{split}
\end{equation}
Hence $\mathscr{S}^{\flat}$ acts on $G$ by (anti-)automorphisms. We must check now that $\mathscr{S}^{\flat} \subset \mathscr{S}$ is a subgroup.
\begin{equation*}
\begin{split}
&\left( b_2\cdot b_1 \right)_s\left( \beta \circ \alpha \right) = \left( b_2 \right)_s\left( t\left( \left( b_1 \right)_s \left( \beta \circ \alpha \right) \right) \right) \circ_V \left( b_1 \right)_s \left( \beta\circ \alpha \right)  = \\
&=\left( b_2 \right)_s \left( \beta' \circ \alpha' \right) \circ_V \left( b_1 \right)_s \left( \beta \circ \alpha \right)= \\
&=\left( \left( b_2 \right)_s \left( \beta' \right) \circ_H \left( b_2 \right)_s \left( \alpha' \right) \right) \circ_V \left( \left( b_1 \right)_s \left( \beta \right) \circ_H \left( b_1 \right)_s \left( \alpha \right) \right) = \\
&=\left( \left( b_2 \right)_s \left( \beta' \right) \circ_V \left( b_1 \right)_s \left( \beta \right) \right) \circ_H \left( \left( b_2 \right)_s \left( \alpha' \right) \circ_V \left( b_1 \right)_s \left( \alpha \right) \right)= \\
&=\left( b_2\cdot b_1 \right)_s \left( \beta \right) \circ_H \left( b_2\cdot b_1 \right)_s \left( \alpha \right)\,,
\end{split}
\end{equation*}
where in the last identity we have used the exchange identity of the symmetroid $S\,\rightrightarrows\,G\,\rightrightarrows\,\Omega$. In a diagramatic way the previous identity can be proven according to the following steps. Firstly
\begin{equation*}
\begin{tikzcd}[column sep=small, row sep=small]
 & \scriptstyle \beta \arrow[d, "b_1(\beta)", Rightarrow] \arrow[dr, bend left = 30] & & \scriptstyle \alpha \arrow[dr, bend left = 30] \arrow[d, "b_1(\alpha)", Rightarrow] &  \\
\scriptstyle \bullet \arrow[r, dash] \arrow[ur, dash, bend left = 30] \arrow[dr, dash, bend right = 30] & \scriptstyle \beta' \arrow[d, "b_2(\beta')", Rightarrow] \arrow[r] & \scriptstyle \bullet \arrow[ur, dash, bend left = 30] \arrow[r, dash] \arrow[dr, dash, bend right = 30] & \scriptstyle \alpha' \arrow[d, "b_2(\alpha')", Rightarrow] \arrow[r] & \scriptstyle \bullet \\
& \scriptstyle \beta'' \arrow[ur, bend right = 30] & & \scriptstyle \alpha'' \arrow[ur, bend right = 30]
\end{tikzcd} \simeq  \begin{tikzcd} [column sep=large, row sep=scriptsize]
 & \beta \circ \alpha \arrow[d, "b_1(\beta) \circ_H b_1(\alpha)", Rightarrow] \arrow[dr, bend left = 30] & \\
\scriptstyle \bullet \arrow[r, dash] \arrow[ur, dash, bend left = 25] \arrow[dr, dash, bend right = 25] & \scriptstyle \beta' \circ \alpha' \arrow[d, "b_1(\beta') \circ_H b_1(\alpha')", Rightarrow] \arrow[r] & \scriptstyle \bullet \\
& \beta'' \circ \alpha'' \arrow[ur, bend right = 30] &
\end{tikzcd} \simeq
\end{equation*}
\begin{equation*}
\simeq \begin{tikzcd}[column sep=scriptsize, row sep=small]
& \scriptstyle \beta \circ \alpha \arrow[d, "b_1(\beta\circ \alpha)", Rightarrow] \arrow[dr, bend left = 30] & \\
\scriptstyle \bullet \arrow[r, dash] \arrow[ur, dash, bend left = 30] \arrow[dr, dash, bend right = 30] & \scriptstyle \beta' \circ \alpha' \arrow[d, "b_2(\beta' \circ \alpha')", Rightarrow] \arrow[r] & \scriptstyle \bullet \\
& \scriptstyle \beta'' \circ \alpha'' \arrow[ur, bend right = 30] &
\end{tikzcd} \simeq  \begin{tikzcd}[column sep=scriptsize, row sep=small]
 & \scriptstyle \beta \circ \alpha \arrow[dd, "b_2 \circ_V b_1" near start, " (\beta \circ \alpha)" near end, Rightarrow] \arrow[dr, bend left = 30] & \\
\scriptstyle \bullet \arrow[ur, dash, bend left = 30] \arrow[dr, dash, bend right = 30] &  & \scriptstyle \bullet \\
& \scriptstyle \beta'' \circ \alpha'' \arrow[ur, bend right = 30] &
\end{tikzcd}
\end{equation*}
The first diagram of the above chain of equalities is equal to
\begin{equation}
\begin{tikzcd}
 & \beta \arrow[dd, "b_2 \circ_V b_1(\beta)", Rightarrow] \arrow[dr, bend left = 30] &  & \alpha \arrow[dr, bend left = 30] \arrow[dd, "b_2 \circ_V b_1(\alpha)", Rightarrow] & \\
\bullet \arrow[ur, dash, bend left = 30] \arrow[dr, dash, bend right = 30] &  & \bullet \arrow[ur, dash, bend left = 30] \arrow[dr, dash, bend right = 30] & & \bullet  \\
& \beta'' \arrow[ur, bend right = 30] & & \alpha'' \arrow[ur, bend right = 30]
\end{tikzcd} \simeq \begin{tikzcd}[column sep=large]
 & \beta \circ \alpha \arrow[dd, "(b_2 \circ_V b_1) (\beta) \circ_H " near start, "\circ_H (b_2 \circ_V b_1 (\alpha))" near end, Rightarrow] \arrow[dr, bend left = 30] & \\
\bullet \arrow[ur, dash, bend left = 30] \arrow[dr, dash, bend right = 30] &  & \bullet \\
& \beta'' \circ \alpha'' \arrow[ur, bend right = 30] &
\end{tikzcd}
\end{equation}
Therefore
\begin{equation}
\begin{tikzcd}
 &  \beta \circ \alpha \arrow[dd, "b_2 \circ_V b_1" near start, " (\beta \circ \alpha)" near end, Rightarrow] \arrow[dr, bend left = 30] & \\
\bullet \arrow[ur, dash, bend left = 30] \arrow[dr, dash, bend right = 30] &  & \scriptstyle \bullet \\
& \beta'' \circ \alpha'' \arrow[ur, bend right = 30] &
\end{tikzcd} \,=\,\begin{tikzcd}[column sep=large]
 & \beta \circ \alpha \arrow[dd, "(b_2 \circ_V b_1) (\beta) \circ_H " near start, "\circ_H (b_2 \circ_V b_1 (\alpha))" near end, Rightarrow] \arrow[dr, bend left = 30] & \\
\bullet \arrow[ur, dash, bend left = 30] \arrow[dr, dash, bend right = 30] &  & \bullet \\
& \beta'' \circ \alpha'' \arrow[ur, bend right = 30] &
\end{tikzcd}
\end{equation}
\end{proof}

According to the proposition above, for every $b\in \mathscr{S}$ the map $\varphi_b$ determines a covariant or contravariant functor, that is a homomorphism or anti-homomorphism, of the groupoid $G$ on itself.
In general, given two functors $F_1,F_2 :G_1\rightarrow G_2$ between two groupoids, a natural transformation (see \cite{I-R-2019} for more details, whereas \cite{michor} deals with applications of natural transformations in differential geometry)
$$
\Phi \,\colon\, F_1\,\Rightarrow\, F_2
$$
is an assignment of a morphism in the groupoid $G_2$, say $\Phi(x)\,\colon\, F_1(x)\,\rightarrow\,F_2(x)$, to any object $x\in G_1$, such that
\begin{equation}
\Phi(y)\circ F_{1}(\alpha) = F_{2}(\alpha)\circ \Phi(x),\quad \forall \alpha \in G_1(x,y)\,.
\end{equation}
If every $\Phi(x)$ is an isomorphism (this happens in the case of groupoids), the natural transformation is an equivalence and there exists another natural transformation $\Psi\,\colon\, F_2\,\Rightarrow\, F_1$ defined by $\Psi(x) = \Phi(x)^{-1}$ which is the inverse of $\Phi$.

Among the flat bisections, the identity $b_e \in \mathscr{S}^{\flat}$ determines the identity automorphism of the groupoid $G\rightrightarrows \Omega$, denoted by $\mathbf{1}_G = \varphi_{b_e}$. If there is a natural transformation $\Phi\,\colon\,\varphi_{b_e} \,\rightarrow\,\varphi_b$ between the functor $\varphi_b\,\colon\,G\,\rightarrow\,G$ and $\mathbf{1}_G \,\colon\,G\,\rightarrow\,G$, there exists a morphism $\Phi(x)\,\colon\,x\,\rightarrow\,\varphi_b(x)$ in $G$ such that:
\begin{equation}
\varphi_b(\alpha) = \Phi(y)\circ \alpha \circ (\Phi(x))^{-1}\,.
\end{equation}
Since $\varphi_b$ is a groupoid automorphism, we have that
\begin{equation}
\varphi_b(1_x)= \Phi(x)\circ 1_x \circ (\Phi(x))^{-1} = 1_{t(\Phi(x))}
\end{equation}
where the map $t\circ \Phi\,\colon\,\Omega\,\rightarrow\,\Omega$ must be bijective. This means that the graph of the map $\Phi$ determines a bisection of the groupoid $G\,\rightrightarrows\, \Omega$.

The set of flat bisections determining groupoid automorphisms which are naturally equivalent to the identical automorphism defines a subgroup, $\mathscr{S}^{\flat}_0$ called the group of inner symmetries associated with the symmetroid $S\,\rightrightarrows \,G$.

Now, let us recall that if $G\,\rightrightarrows \,\Omega$ is the groupoid describing the kinematical structure of a quantum system, we called the $2$-groupoid over $G\,\rightrightarrows\,\Omega$ defined by the action of the groupoid $G$ onto itself (on the left, right and inversions) the canonical symmetroid of $G\,\rightrightarrows\,\Omega$ and we denoted it by $S(G)\,\rightrightarrows\,G\,\rightrightarrows\,\Omega$.

The group of flat bisections of $S(G)$ defines a subgroup $\mathscr{S}(G)^{\flat}$ of the group of (anti-)automorphisms of $G\,\rightrightarrows\,\Omega$ called the canonical group of symmetries of $G\,\rightrightarrows\,\Omega$.
The elements of this group play the role of structural symmetries  of the system in the groupoid framework.
The following Th. \ref{thm: dynamical symmetries inside kinematical ones} shows that every inner dynamical symmetry, which is described as an element of $\mathscr{S}_{0}^{\flat}$, is realized as an element of $\mathscr{S}(G)^{\flat}$.

\begin{theorem}\label{thm: dynamical symmetries inside kinematical ones}
Let $S\,\rightrightarrows\,G\,\rightrightarrows\,\Omega$ be a symmetroid describing the microscopic symmetries of a quantum system whose kinematical structure is given by the groupoid $G\,\rightrightarrows\,\Omega$. Then the group of inner symmetries, $\mathscr{S}_0^{\flat}$ associated with $S\,\rightrightarrows\,G$ is a subgroup of $\mathscr{S}(G)^{\flat}$, i.e., the global inner symmetries of a quantum systems are elements of the canonical group of symmetries of $G\,\rightrightarrows\,\Omega$.
\end{theorem}
\begin{proof}
The proof of this theorem is obtained by showing that any bisection whose associated injective function $(b_{\Phi})_s\,\colon\,G\,\rightarrow\,S$ reads
\begin{equation}
(b_{\Phi})_s(\alpha)\,\colon\,\alpha\,\Rightarrow \Phi(y)\circ \,\alpha \,\circ (\Phi(x))^{-1}\,,
\end{equation}
belongs to the group $\mathscr{S}(G)^{\flat}$. In the formula above $\alpha\,\colon\,x\,\rightarrow\,y$ is a transition in $G$ and $\Phi\,\colon\,\Omega\,\rightarrow\,G$ is the map associated with a natural transformation.

Indeed, the maps:
\begin{equation}
\begin{split}
(b^R_{\Phi})_s(\alpha)\,\colon\,\alpha\,&\Rightarrow\,\alpha\circ(\Phi(x))^{-1}\\
(b^L_{\Phi})_s(\alpha)\,\colon\,\alpha\,&\Rightarrow\,\Phi(y)\circ\alpha
\end{split}
\end{equation}
defines two bisections of the canonical groupoid associated with $G$. It is immediate to see that the bisection $b_{\Phi}$ can be written as the product $b_{\Phi}^L\cdot b_{\Phi}^R$, which ends the proof.
\end{proof}

The previous theorem lifts to the abstract setting of groupoids the celebrated Wigner's theorem on the realization of symmetries of quantum systems, i.e., under the assumptions that the family of substitutions describing the symmmetry of a system must be a connected $2$-groupoid, then its associated group of inner symmetries must define a subgroup of the group of inner (anti-)automorphisms of the groupoid of the system, where inner is used in analogy with the inner group of automorphism of a group.

\section{Examples.}

We will end the paper with some examples which will illustrate how to construct symmetroids out of specific groupoids. We will consider finite groupoids in order to avoid technical difficulties which would hide the algebraic aspects of the construction which we have focused on.

\subsection{The canonical symmetroid of the groupoid of pairs $\mathbf{G}(\mathbf{\Omega}_n)$}\label{can_sym_pair}

Let $\Omega_n$ be the finite set of $n$ elements. Let us consider the groupoid of pairs $G(\Omega_n)\rightrightarrows \Omega_n$. As already noticed in the discussion after the Ex.\ref{ex.reconstruction}, the group of bisections of this groupoid is made up of the graphs of the bijective maps over $\Omega_n$.

The canonical symmetroid $S(G(\Omega_n))\,\rightrightarrows\, G(\Omega_n)\,\rightrightarrows\,\Omega_n$ is made up of substitutions
\begin{equation*}
\begin{split}
\xi^R_{\gamma}\,\colon\,\alpha\,\Rightarrow\,\alpha\circ \gamma^{-1}\\
\xi^L_{\gamma}\,\colon\,\beta\,\Rightarrow\,\gamma\circ\beta\\
\tau_{\gamma}\,\colon\,\gamma\,\Rightarrow\,\gamma^{-1}\,,
\end{split}
\end{equation*}
where
\begin{equation*}
\alpha\,\colon\,x\,\rightarrow\,z,\quad \beta \,\colon\,z'\,\rightarrow\,x,\quad \gamma\,\colon\,x\,\rightarrow\,y\,.
\end{equation*}
A straightforward computation shows that the following commutation relations hold:
\begin{equation}
\begin{split}
&\xi^R_{\gamma_1} \circ_V \xi^L_{\gamma_2} = \xi^L_{\gamma_2} \circ_V \xi^R_{\gamma_1}\\
&\xi^L_{\gamma}\circ_V \tau_{\alpha} = \tau_{\alpha\circ \gamma^{-1}}\circ_V \xi^R_{\gamma}\\
&\xi^R_{\delta}\circ_V \tau_{\alpha} = \tau_{\delta\circ\alpha}\circ_V \xi^L_{\delta}
\end{split}
\end{equation}
where $\delta\,\colon\,z\,\rightarrow\,w$ and $\gamma_1,\:\gamma_2$ are generic transitions in $G(\Omega_n)$. Due to these properties, any element of the canonical symmetroid can be rearranged in such a way that we have the composition of a left multiplication, a right multiplication and an inversion. Then, the group of bisections $\mathscr{S}(G(\Omega_n))$ is generated by the graphs of functions which assume the following expressions:
\begin{equation}
\begin{split}
(b^L_{\Phi})_s((y,x))\,\colon\,(y,x)\,&\Rightarrow \,(\Phi(y),y)\circ(y,x) = (\Phi(y),x)\\
(b^R_{\Phi})_s((y,x))\,\colon\,(y,x)\,&\Rightarrow \,(y,x)\circ((\Phi(x),x)^{-1})=(y,\Phi(x))\\
(b^{\tau})_s((y,x)) \,\colon\,(y,x)\,&\Rightarrow\,(x,y)\,,
\end{split}
\end{equation}
where $\Phi\,\colon\,\Omega_n\,\rightarrow\,\Omega_n$ is a bijective map over $\Omega_n$. In particular, the bisection $b^{\tau}$ is the groupoid anti-automorphism which maps $G(\Omega_n)$ to its opposite groupoid, the arrows of which are all reversed.

On the other hand, the subgroup $\mathscr{S}(G(\Omega_n))^{\flat}$ of flat bisections of the canonical symmetroid is generated by the graphs of the following functions:
\begin{equation}
(b^{Ad}_{\Phi})_s((y,x))\,\colon\,(y,x)\,\Rightarrow\, (\Phi(y),y)\circ(y,x)\circ((\Phi(x),x)^{-1}) = (\Phi(y),\Phi(x))\,.
\end{equation}
This is the group of canonical symmetries of the groupoid of pairs $G(\Omega_n)$.

\subsection{The canonical symmetroid of an Action Groupoid.}

In this second example we will consider a finite action groupoid $\mathbf{G}\rightrightarrows \Omega_n$ over the set of the first $n$ natural numbers (see \cite{landsman, lie-groupoids, I-R-2019} for more details about action groupoids). For the sake of simplicity we will consider the cyclic group $\Gamma$ of order $n$ generated by the element $r$, i.e.,
\begin{equation*}
\Gamma = \left\lbrace e,\, r,\, r^2,\, \cdots \,,\,r^{n-1}  \right\rbrace\,.
\end{equation*}
The action of $\Gamma$ on $\Omega_n$ is generated by the bijective map $\mu_{r}\,\colon\,\Omega_n\,\rightarrow\,\Omega_n$ which is expressed as follows:
\begin{equation}
\mu_r(j) = \left[ j+1 \right]_n\,,
\end{equation}
where the symbol $\left[ j+1 \right]_n$ denotes the number $j+1$ modulo $n$. This action is free and transitive and a transition $\alpha\in \mathbf{G} = \Gamma \times \Omega_n$ is represented via the triple
\begin{equation}
\alpha = \left( \mu_g(j), \, g,\,j \right)\,.
\end{equation}
Since the action is free and transitive the isotropy groups are trivial and the groupoid is connected (this groupoid is actually isomorphic to the pair groupoid $G(\Omega_n)\,\rightrightarrows\,\Omega_n$). Any bisection of this groupoid is the graph of a bijective map of the unique orbit, $\mathscr{O}=\Omega_n$, of the action of the group $\Gamma$ on $\Omega_n$ via the map $\mu_{\cdot}$. Therefore, a bisection $b\in \mathscr{G}$ is represented by the map
\begin{equation}
\left( b_{\Phi}\right)_s (j) = (\Phi(j), g, j)\,,
\end{equation}
where $g\in \Gamma$ is the unique group element connecting $j$ and $\Phi(j)$. As we have already seen in example \ref{can_sym_pair}, the canonical symmetroid $S(\mathbf{G})\rightrightarrows\mathbf{G}\rightrightarrows \Omega_n$ is generated by the composition of substitutions of the type:
\begin{equation}
\xi^R_{\alpha},\; \xi^L_{\alpha},\;\tau_{\alpha}\,.
\end{equation}
Then, the associated group of bisections, $\mathscr{S}$, is given by the right (left) composition $b^{R}_{\Phi} (b^L_{\Phi})$ with a bisection of the groupoid $\mathbf{G}$, say $\Phi \in \mathscr{G}$, and by the inversion transformation $b^{\tau}$. The subgroup of canonical symmetries, i.e., the subgroup of flat bisections, contains bisections which are represented by the map
\begin{equation}
\begin{split}
&(b_{\Phi}^{Ad})_s((\mu_g(j), \, g,\, j)) \,\colon\, (\mu_g(j), \, g,\, j)\,\Rightarrow \\
\Rightarrow (\left(\Phi  \circ\mu_g\right)&(j),\,\gamma,\,\mu_g(j))\circ (\mu_g(j),\, g,\,j )\circ (j,\, \gamma^{-1}, \Phi(j))\,= \\
& = \left( \left(\Phi  \circ\mu_g\right)(j),\, \gamma g \gamma^{-1},\, \Phi(j) \right)
\end{split}
\end{equation}
In particular, the action on a unit $1_j$ provides the following result
\begin{equation}
(b_{\Phi}^{Ad})_s((j,\,e,\,j))\,\colon\,(j,\,e,\,j) \,\Rightarrow \, (\Phi(j), \, e,\, \Phi(j))\,,
\end{equation}
which clearly shows the fact that the isotropy normal subgroupoid is preserved by the action of the group $\mathscr{S}^{\flat}$.

\subsection{The swap symmetroid}
In this last example we are going to present the symmetroid associated with the direct product of two isomorphic groupoids and which we call the swap symmetroid since it should describe swapping transformations between two identical quantum systems. For the sake of simplicity we consider the direct product of two groupoids of pairs over the finite set made up of only 2 elements, i.e.,
\begin{equation*}
\mathbf{G} = G(\Omega_2)\times G(\Omega_2)\rightrightarrows \Omega_2\times\Omega_2 =:\mathbf{\Omega}\,.
\end{equation*}
This groupoid contains the following transitions:
\begin{equation*}
\begin{split}
(\alpha_a, \alpha_b)&,\;\;(\alpha_a,\alpha^{-1}_b), \;\;(\alpha^{-1}_a,\alpha_b),\;\;(\alpha^{-1}_a,\alpha^{-1}_b),\;\; (\alpha_a, 1_{+_b}), \;\;(\alpha_a, 1_{-_b}) \\
(\alpha^{-1}_a,1_{+_a}),&\;\; (\alpha^{-1}_a,1_{-_b}),\;\;(1_{+_a},\alpha_b),\;\; (1_{+_a},\alpha^{-1}_b),\;\; (1_{-_a},\alpha_b),\;\; (1_{-_a},\alpha^{-1}_b)\\
& (1_{+_a},1_{+_b}),\;\; (1_{+_a},1_{-_b}),\;\; (1_{-_a},1_{+_b}),\;\; (1_{-_a},1_{-_b})\,,
\end{split}
\end{equation*}
with the obvious composition rule, where
\begin{equation*}
\alpha_a\,\colon\,-_{a}\,\rightarrow\,+_{a}\,,\quad \alpha_{b}\,\colon\, -_{b}\,\rightarrow\,+_{b} \,.
\end{equation*}
Let us consider, now, the following set of non trivial substitutions:
\begin{equation*}
\begin{split}
\xi^{\sigma}_{(\alpha,+)}\,\colon\, (\alpha_a,1_{+_b})\,\Rightarrow \, (1_{+_a},\alpha_b) \quad \xi^{\sigma}_{(\alpha,-)}\,\colon\, (\alpha_a,1_{-_b})\,\Rightarrow \, (1_{-_a},\alpha_b) \\
\xi^{\sigma}_{(\alpha^{-1},+)}\,\colon\, (\alpha^{-1}_a,1_{+_b})\,\Rightarrow \, (1_{+_a},\alpha^{-1}_b) \quad
\xi^{\sigma}_{(\alpha^{-1},-)}\,\colon\, (\alpha^{-1}_a,1_{-_b})\,\Rightarrow \, (1_{-_a},\alpha^{-1}_b)\\
\xi^{\sigma}_{(\alpha,\alpha^{-1})}\,\colon\, (\alpha_a,\alpha^{-1}_b)\,\Rightarrow \, (\alpha^{-1}_a,\alpha_b)\quad
\xi^{\sigma}_{(+,-)}\,\colon\, (1_{+_a},1_{-_b})\,\Rightarrow \, (1_{-_a},1_{+_b})
\end{split}
\end{equation*}
along with their inverse substitutions and the trivial ones, which will be denoted $\xi^{e}_{\cdot}$. The source and target maps, $s_1,t_1$, are the obvious ones, and the vertical composition of two substitutions $\xi^{\sigma}_j\circ_V\xi^{\sigma}_i$ is defined whenever $s_1\left( \xi^{\sigma}_j \right)= t_1\left( \xi^{\sigma}_i \right)$. This implies, for instance, that each non trivial substitution is composable only with its inverse or with the suitable unit substitutions.

Straightforward computations, like the following one
\begin{equation*}
\begin{split}
\left( \xi^{\sigma}_{(\alpha,+)} \circ_V \left( \xi^{\sigma}_{(\alpha,+)} \right)^{-1} \right) &\circ_H \left( \xi^{\sigma}_{(\alpha^{-1},+)} \circ_V \left( \xi^{\sigma}_{(\alpha^{-1},+)} \right)^{-1} \right) = \\
=\xi^e_{(\alpha,+)} &\circ_H \xi^{e}_{(\alpha^{-1},+)} = \xi^e_{(+,+)} = \\
= \left( \xi^{\sigma}_{(\alpha,+)} \circ_H \xi^{\sigma}_{(\alpha^{-1},+)} \right) &\circ_V \left( \left( \xi^{\sigma}_{(\alpha,+)}\right)^{-1} \circ_H \left( \xi^{\sigma}_{(\alpha^{-1},+)} \right)^{-1} \right) = \xi^e_{(+,+)}\circ_V \xi^{e}_{(+,+)}\,,
\end{split}
\end{equation*}
show that the exchange identity is satisfied and the above set of substitutions determines a symmetroid $S\,\rightrightarrows\,\mathbf{G}\,\rightrightarrows \,\mathbf{\Omega}$.

As already seen, a bisection $b\subset S $ determines an injective map $b_s\,\colon\,\mathbf{G}\,\rightarrow \,S$ such that $t_1\circ b_s$ is a bijective map on $\mathbf{G}$. In this case it is possible to associate with a transition in $\mathbf{G}$, for instance $(\alpha_a, 1_{+_b})$, up to two substitutions, either $\xi^{\sigma}_{(\alpha, \, +)}$ (if possible) or $\xi^{e}_{(\alpha, \, +)}$. A bisection $b$ is flat if the corresponding function $b_s$ sastisfies the property in Def.\ref{flat_bisection}. Recalling that the transition $(\alpha_a, 1_{+_b})\in \mathbf{G}$ is composable with the transitions
\begin{equation}\label{ex:composable_transitions}
(1_{-_a},1_{+_b}),\: (\alpha_a^{-1},1_{+_b}),\: (1_{-_a}, \alpha_b),\:  (\alpha_a^{-1},\alpha_b),\: (1_{+_a}, \alpha_b^{-1})\,,
\end{equation}
if $b_s((\alpha, +)) = \xi^{\sigma}_{(\alpha,+)}$ then a flat bisection must associate a non trivial substitution to all the other transitions which are composable with it (let us note that the last transition in Eq.(\ref{ex:composable_transitions}) is composable on the left of the transition $(\alpha_a,\,1_{+_b})$, while the others on the right). Viceversa, if $b_s((\alpha, +)) = \xi^{e}_{(\alpha,+)}$, all the other images must be trivial bisections in order to satisfy property \ref{flat_bisection}. This implies that the group of flat bisections $\mathscr{S}^{\flat}$ contains only two elements, i.e.,
$$
\mathscr{S}^{\flat} = \left\lbrace b_e , b_{\sigma} \right\rbrace
$$
such that $b_{\sigma}\cdot b_{\sigma} = b_e$. Let us consider, now, the map $\Phi\,\colon\,\mathbf{\Omega}\,\rightarrow \mathbf{G}$ such that
\begin{equation*}
\begin{split}
\Phi((+_a,+_b)) = (1_{+_a} , 1_{+_b}),\;\; \Phi((+_a,-_b)) = (\alpha^{-1}_a,\alpha_b) \\
\Phi((-_a,+_b)) = (\alpha_a,\alpha^{-1}_b),\;\;\Phi(-_a,-_b) = (1_{-_a},1_{-_b})\,.
\end{split}
\end{equation*}
A direct computation shows that the automorphism $\varphi_{\sigma}$ associated with $b_{\sigma}$ is naturally equivalent to the identity, the natural transformation being determined by the above map $\Phi$. This means that $\mathscr{S}^{\flat}=\mathscr{S}^{\flat}_0$ and it is a subgroup of the canonical group of symmetries of the symmetroid $S(\mathbf{G})\rightrightarrows \mathbf{G}\rightrightarrows \mathbf{\Omega}$.

\section{Conclusions.}
In this paper we analysed how to describe symmetries of (quantum) systems in the groupoid approach to Schwinger's picture of Quantum Mechanics. Since a quantum system is described by a groupoid $G\,\rightrightarrows \Omega$, the first possible method to define a symmetry is microscopic: one considers a family of substitutions that exchange the transitions in $G$ among themselves. A detailed list of axioms that these substitution must satisfy is provided and it is shown that the algebraic structure suited to this microscopic point of view is that of a 2-groupoid, say $S\rightrightarrows G \rightrightarrows \Omega$.

On the other hand, symmetries of physical systems are usually implemented via groups of transformations. This approach to the description of symmetries is recovered in the groupoid picture of Quantum Mechanics considering the notion of bisection: the set $\mathscr{S}$ of bisections of the symmetroid $S\rightrightarrows G$ is a group, and the subgroup, $\mathscr{S}^{\flat}$, of flat bisections, which are those respecting the composition law of $G$, acts on $G$ via groupoid (anti-)automorphisms. Moreover the relation between groupoids and its group of bisections is even deeper, since it has been shown that a groupoid can be reconstructed from the group of bisections and its space of objects. This global point of view also allows to present an abstract extension of Wigner's theorem within the groupoid approach, where a canonical symmetroid, denoted $S(G)\rightrightarrows G$, emerges as a universal environment containing every other group of symmetries of the system.

The main focus of this work has been on the algebraic aspects connected to the theory of 2-groupoids and their groups of bisections. This is clearly only the first step towards a more complete analysis of symmetries in Schwinger's picture of Quantum Mechanics, and future efforts will be devoted to the extension of the results presented in this paper to the realm of continuous and Lie groupoids \cite{lie-groupoids}. Moreover, the theory illustrated here, should be extended to include also symmetries of the algebra of observables. A preliminary result along this research line, consists in the following fact: The set of characteristic functions supported on the bisections of a groupoid $G\rightrightarrows \Omega$ is a subset of its groupoid algebra. It can be easily proven, at least in the case of discrete groupoids, that these functions close on a unitary representation of the group of bisections $\mathscr{G}$ of the groupoid $G$. However, this preliminary activity requires a more systematic approach which will be postponed to forthcoming works.

\section*{Acknowledgments} F.D.C. and A.I. would like to thank partial support provided by the MINECO research project MTM2017-84098-P and QUITEMAD++, S2018/TCS-A4342. A.I. and G.M. acknowledge financial support from the Spanish Ministry of Economy and Competitiveness, through the Severo Ochoa Programme for Centres of Excellence in RD(SEV-2015/0554). G.M. would like to thank the support provided by the Santander/UC3M Excellence Chair Programme 2019/2020, and he is also a member of the Gruppo Nazionale di Fisica Matematica (INDAM), Italy. L.S. would like to thank the support provided by Italian MIUR through the Ph.D. Fellowship at Dipartimento di Matematica R.Caccioppoli.

\end{document}